\documentclass{article}

\usepackage[utf8]{inputenc} 
\usepackage[T1]{fontenc}    
\usepackage{hyperref}       
\usepackage{url}            
\usepackage{booktabs}       
\usepackage{amsfonts}       
\usepackage{nicefrac}       
\usepackage{microtype}      
\usepackage{lipsum}
\usepackage{fancyhdr}       
\usepackage{graphicx}       
\graphicspath{{media/}}     

\usepackage{lineno,hyperref}
\modulolinenumbers[5]
\usepackage{geometry} 
\usepackage{amsmath}  
\usepackage{amsfonts}
\usepackage{forest} 
\usepackage{graphicx}  
\usepackage{amssymb}
\usepackage{bbm}
\usepackage{setspace}

\usepackage{amsthm} 

\newtheorem{theorem}{Theorem}
\newtheorem{prop}{Proposition}
\newtheorem{lemma}{Lemma}

\usepackage{enumitem}

\title{A Countable-Type Branching Process Model for the Tug-of-War Cancer Cell Dynamics}

\author{
  Ren-Yi Wang$^{1}$ and Marek Kimmel$^{2}$ \\
  $^{1,2}$\text{Department of Statistics, Rice University, Houston, TX, 77005, USA} \\
  $^{2}$\text{Department of Systems Biology and Engineering, Silesian University of Technology,}\\ \text{ Akademicka 16, 44-100 Gliwice, Poland} \\
  $^{1}$\texttt{rw47@rice.edu} \\
}

\begin{document}
\maketitle

\begin{abstract}
We consider a time-continuous Markov branching process of proliferating cells with a countable collection of types. Among-type transitions are inspired by the Tug-of-War process introduced by \cite{mcfarland2014tug} as a mathematical model for competition of advantageous driver mutations and deleterious passenger mutations in cancer cells. We introduce a version of the model in which a driver mutation pushes the type of the cell $L$-units up, while a passenger mutation pulls it $1$-unit down. The distribution of time to divisions depends on the type (fitness) of cell, which is an integer. The extinction probability given any initial cell type is strictly less than $1$, which allows us to investigate the transition between types (type transition) in an infinitely long cell lineage of cells. The analysis leads to the result that under driver dominance, the type transition process escapes to infinity, while under passenger dominance, it leads to a limit distribution. Implications in cancer cell dynamics and population genetics are discussed.
\end{abstract}


\section{Introduction}
Cancerous tumors are believed to be most frequently initiated by a single mutated cell that generates a population of cells (a tumor) with elevated mutation rates and genome instability. Mutations are categorized into two types: driver and passenger mutations. Driver mutations confer selective advantages to the cell by increasing its fitness and passenger mutations were viewed as neutral. Recent theoretical and experimental studies challenge the hypothesis of neutrality of passenger mutations. Theory in \cite{mcfarland2014tug} and experiments in \cite{mcfarland2017damaging} suggest passenger mutations are mildly deleterious and they inhibit tumor growth and reduce metastasis when passenger load exceeds a threshold.
	
In this paper, we analyze and simulate a branching process model with countably many types inspired by the original model in \cite{mcfarland2014tug} that captures the experimental outcomes in \cite{mcfarland2017damaging} but is simple enough for rigorous analysis. In our model, rates for driver and passenger mutations are $\mu$ and $\nu$, respectively. The type of a cell is a monotonic function of its fitness, which is integer-valued. Given any initial cell type, we show global extinction probabilities coincides with partial extinction probabilities, and they are all strictly less than $1$. The fact that extinction does not always occur allows us to investigate the transition between cell types along a non-extinct lineage, which provides insights into the average fitness of the cell population as well as the interaction between negative selection and deleterious passenger mutations. We discover a dichotomy of the existence of a limiting distribution based on the inequality $\nu \leq \mu L$, where $L$ depends on the selection coefficients for driver and passenger mutations. When $\nu$ does not exceed the threshold $\mu L$, the type transition process is transient and a limiting distribution does not exist. In this case, we say the process is driver dominant. On the other hand, when $\nu$ exceeds the threshold $\mu L$, type transition process admits a limiting distribution and we say the process is passenger dominant. 

\section{Model Setup} 
In this section, we present the derivation of our model and motivate simplifications that lead to the form of the model we analyze. Mathematical details regarding construction of the model are deferred to Appendix $\ref{model_construction}$.

\subsection{Continuous-time Markov Chain} \label{section 2.1}
The birth rate (fitness) and death rate of a cell with $n_{d}$ drivers and $n_{p}$ passengers are

\begin{align*}
    b(n_{d},n_{p}) = b_{0}\frac{(1+s_{d})^{n_{d}}}{(1+s_{p})^{n_{p}}},\;\; d(n_{d},n_{p}) = d.
\end{align*}
    
\noindent The tumor starts from a single cancer cell with birth rate $b_{0} > d$ (type $(0,0)$ cell), where $d$ is assumed to be the death rate of the cell. Without loss of generality, we may take $d = 1$. $s_{d}$ and $s_{p}$ are selection coefficients for driver and passenger mutations and since passenger mutations are only mildly deleterious, $s_{d} > s_{p}$. Let $\mu$ denote the driver mutation rate and $\nu$ denote the passenger mutation rate $(\mu < \nu)$. It appears reasonable to assume $d > \nu$ since mutations are rare events. Numerical values for mutation rate and death rate supporting the assumption $d>\nu$ can be found in \cite{durrett2015branching} and \cite{mcfarland2014tug}. The evolution of tumor cell population is modeled as a Markov branching process with type space $\mathbb{Z}_{+}\times\mathbb{Z}_{+}$. The lifetime of a type $(n_{d},n_{p})$ cell is exponentially distributed with rate $\delta(n_{d},n_{p})=b(n_{d},n_{p})+\mu+\nu+d$. At the end of the cell's lifetime, it can take four actions: division (producing two $(n_{d},n_{p})$ cells), acquiring a driver mutation (producing a $(n_{d}+1,n_{p})$ cell), acquiring a passenger mutation (producing a $(n_{d},n_{p}+1)$ cell), and death. The four actions occur with probabilities $b(n_{d},n_{p})/\delta(n_{d},n_{p}),\mu/\delta(n_{d},n_{p}),\nu/\delta(n_{d},n_{p})$, and $d/\delta(n_{d},n_{p})$, respectively. The proliferation scheme for the Markov branching process is presented below with the rates of four actions. \\
    
\begin{align*}
&\begin{forest}
[$(n_{d}\mid n_{p})$
    [$(n_{d}\mid n_{p})$]
    [$(n_{d}\mid n_{p})$]
]
\end{forest}
&&\begin{forest}
[$(n_{d}\mid n_{p})$
    [$(n_{d}+1\mid n_{p})$]
]
\end{forest}
&&\begin{forest}
[$(n_{d}\mid n_{p})$
    [$(n_{d}\mid n_{p}+1)$]
]
\end{forest}
&&\begin{forest}
[$(n_{d}\mid n_{p})$
    [$\emptyset$]
]
\end{forest} \\
\text{Rate: } &\quad\qquad b(n_{d},n_{p}) &&\quad\qquad \mu &&\quad\qquad \nu &&\qquad d
\end{align*}

\subsubsection{Reduced Process: Driver-Passenger Relation} \label{reduction}
In the remainder, we will analyze the process under an additional hypothesis that a single driver mutation's effect on fitness can be cancelled by $L$ passenger mutations, that is, $1+s_{d} = (1+s_{p})^{L}$, or

\begin{equation}
\begin{aligned}
    \log_{(1+s_{p})}(1+s_{d})=L \in \mathbb{N}, L \geq 2.
\end{aligned} 
\label{eq:dp_relation_man}
\end{equation}

\noindent This simplification seems a step in the right direction since it is always possible to take the floor or ceiling of $\log_{(1+s_{p})}(1+s_{d})$ to obtain an approximation of the process. Therefore, $b(n_{d},n_{p}) = b_{0}(1+s_{p})^{n_{d}L-n_{p}} = b_{0}(1+s_{p})^{i} = b(i)$ and type $i$ is now defined by $i=n_{d}L-n_{p}$. This driver-passenger relation transforms the type space from $\mathbb{Z}_{+} \times \mathbb{Z}_{+}$ to $\mathbb{Z}$ which is easier to work with. Under this simplification, the proliferation scheme becomes

\begin{align*}
&\begin{forest}
[$(i)$
    [$(i)$]
    [$(i)$]
]
\end{forest}
&&\begin{forest}
[$(i)$
    [$(i+L)$]
]
\end{forest}
&&\begin{forest}
[$(i)$
    [$(i-1)$]
]
\end{forest}
&&\begin{forest}
[$(i)$
    [$\emptyset$]
]
\end{forest} \\
\text{Rate: } &\quad b(i) &&\quad \mu &&\quad \nu &&\quad d
\end{align*}

\noindent A mild condition $(\nu + d > \mu L)$ in Lemma \ref{lemma1} guarantees the process is non-explosive. To compute and analyze extinction probabilities, we focus on the embedded branching process in the next section. Corollary of Lemma \ref{lemma1} suggests that under non-explosion, continuous-time process becomes extinct if and only if the embedded process does.

The utility of this simplification will be demonstrated in section \ref{section 3.4}, where we present a dichotomy of existence of a limiting distribution for the type transition process based on the true value of $\nu >\mu L$.

\subsection{Embedded Branching Process}
In the embedded branching process $(\mathbf{E}_{n})_{n\geq 0}$, random lifetimes are replaced by constant time units equivalent to single generations. Let $\delta(i)=b(i)+\mu+\nu+d$; then the probabilities of a type $i$ cell dividing, acquiring a driver, acquiring a passenger, or dying are equal to $b(i)/\delta(i), \mu/\delta(i), \nu/\delta(i)$, and $d/\delta(i)$, respectively.
    
\begin{align*}
&\begin{forest}
[$(i)$
    [$(i)$]
    [$(i)$]
]
\end{forest}
&&\begin{forest}
[$(i)$
    [$(i+L)$]
]
\end{forest}
&&\begin{forest}
[$(i)$
    [$(i-1)$]
]
\end{forest}
&&\begin{forest}
[$(i)$
    [$\emptyset$]
]
\end{forest} \\
\text{Prob: } &\quad\frac{b(i)}{\delta(i)} &&\quad \frac{\mu}{\delta(i)} &&\quad \frac{\nu}{\delta(i)} && \frac{d}{\delta(i)} 
\end{align*}

The mean matrix $M$ has nonzero entries in the $i$th row being

\begin{align*}
    &M_{i,i-1} = \frac{\nu}{\delta(i)}, M_{i,i} = 2\frac{b(i)}{\delta(i)}, M_{i,i+L} = \frac{\mu}{\delta(i)}; \forall i\in\mathbb{Z}.
\end{align*}

\noindent In Proposition $\ref{prop1}$, we use the convergence parameter of $M$ to analyze the relation between global and partial extinction probabilities. 

\section{Results}
    
We present the analysis of the embedded process in this section along with simulations and computations. Table \ref{table:1} contains a list  of parameters for simulations. Note the initial birth rate for the passenger dominance case is higher than that for driver dominance case. This is for simulation purposes, to provide enough surviving lineages. 
    
\begin{table}[!ht]
\centering
\begin{tabular}{|p{1.6cm}|p{.8cm}|p{.8cm}|p{.8cm}|p{.8cm}|p{.8cm}|p{.8cm}|}
    \hline
    \multicolumn{7}{|c|}{Parameter Specifications} \\
    \hline
    Parameters &$b_{0}$ &$\mu$ &$\nu$ &$s_{p}$ &$L$ &$\frac{\nu}{\mu L}$ \\
    \hline 
    Spec. $D1$ &1.1  &0.0251 &0.05 &0.002 &2  &$\leq 1$ \\ 
    Spec. $D2$ &1.1  &0.011  &0.05 &0.002 &5  &$\leq 1$ \\
    Spec. $D3$ &1.1  &0.0051 &0.05 &0.002 &10 &$\leq 1$ \\ 
    Spec. $P1$ &1.15 &0.055  &0.3  &0.002 &2  &$> 1$ \\ 
    Spec. $P2$ &1.15 &0.02   &0.3  &0.002 &5  &$>1$ \\ 
    Spec. $P3$ &1.15 &0.003  &0.3  &0.002 &10 &$>1$ \\
    \hline
\end{tabular}
\caption{Parameter specifications for simulations. The $D$ or $P$ after specification indicates driver or passenger dominance.}
\label{table:1}
\end{table}

\subsection{Extinction Probabilities}

For a branching process with infinitely many types, there are two modes of extinction. Global extinction is the event of the entire population eventually becoming extinct and partial extinction is the event that all types will eventually become extinct. For a precise mathematical definition, define $\mathbf{q}$ as the vector of global extinction probabilities conditional on initial types and $\tilde{\mathbf{q}}$ to be the vector of partial extinction probabilities conditional on initial types. Let $\underline{e}_{i}$ be the bi-infinite vector indexed by $\mathbb{Z}$ whose entries are all zeros except for the $i$th entry being one. Hence, $\{\mathbf{E}_{0}=\underline{e}_{i}\}$ indicates the population is initiated by a type $i$ cell. Let $(\mathbf{E}_{n})_{k}$ denote the number of individuals of type $k$ in the $n$th generation, then
    
\begin{align*}
    &q_{i} = \mathbb{P}(\lim_{n\to\infty}||\mathbf{E}_{n}||_{\ell^{1}}=0 \mid \mathbf{E}_{0}=\underline{e}_{i});\\
    \;\;&\tilde{q}_{i} = \mathbb{P}(\forall k\in\mathbb{Z}, \lim_{n\to\infty}(\mathbf{E}_{n})_{k}=0 \mid \mathbf{E}_{0}=\underline{e}_{i}).
    \end{align*}
    
The probability generating function (PGF) of the progeny distribution is
    
\begin{align*}
    &\mathbf{G}(\mathbf{s})=(\cdots,G_{-1}(\mathbf{s}),G_{0}(\mathbf{s}),G_{1}(\mathbf{s}),\cdots),\\
    &\text{where }G_{i}(\mathbf{s}) = \frac{d}{\delta(i)}+\frac{\nu}{\delta(i)}s_{i-1}+\frac{b(i)}{\delta(i)}s_{i}^{2}+\frac{\mu}{\delta(i)}s_{i+L}.
\end{align*}
    
\noindent By Theorem $3.1$ in \cite{moyal1964multiplicative}, $\mathbf{q}$ is always the minimal non-negative fixed point of the PGF. Hence,
    
\begin{align*}
    q_{i} &= \frac{d}{\delta(i)}+\frac{\nu}{\delta(i)}q_{i-1}+\frac{b(i)}{\delta(i)}q_{i}^{2}+\frac{\mu}{\delta(i)}q_{i+L}.
\end{align*}
    
\noindent As we demonstrate in Theorem \ref{thm1} via a coupling argument, population initiated with a cell with greater fitness is less likely to become extinct, that is, $i>j$ implies $q_{i}\leq q_{j}$. Therefore, $\lim_{i\to\infty}q_{i} = q_{\infty}$ and $\lim_{i\to-\infty}q_{i}=q_{-\infty}$ exist by monotonicity and boundedness. In addition, Proposition \ref{prop1} shows that $\mathbf{q}<\underline{1}$. As a consequence,
	
\begin{align*}
    &q_{\infty} = q_{\infty}^{2} \Rightarrow q_{\infty} = 0 \\
    &q_{-\infty} = \frac{d+\nu q_{-\infty} + \mu q_{-\infty}}{\mu+\nu+d} \Rightarrow q_{-\infty} = 1.
\end{align*}
    
Explicit expression for $q_{i}$ is difficult to find since the difference equation is quadratic, inhomogeneous, and varying for each $i$. We resort to an algorithm in \cite{hautphenne2013extinction} to obtain an approximation of the extinction probability. The algorithm can be applied to the ``doubly''-infinite type space such as the set of all integers $\mathbb{Z}$. Using notation of \cite{hautphenne2013extinction}, let $q_{0}^{(k)}$ be the probability that the process becomes extinct before reaching types in the set $\{i:\;i>k\}\cup\{i:\;i< -k\}$ (taboo types). Let $T_{e}$ be the time of extinction and $\tau_{k}$ be the time of first arrival into the set of taboo types, then 
    
\begin{align*}
    q_{0}^{(k)} = \mathbb{P}(T_{e}<\tau_{k}\mid \mathbf{Z}(0)=\underline{e}_{0}).
\end{align*}
    
\noindent It holds that $\lim_{k\to\infty}\tau_{k}=\infty$ almost surely and Lemma $3.1$ in \cite{hautphenne2013extinction} shows that $q_{0}^{(k)}$ can be used to approximate the extinction probability since $q_{0}^{(k)} \to q_{0}$ as $k\to\infty$. Analogously, $q_{i}$'s can be computed by shifting the set of taboo types. 

Figure \ref{extinction probabilities} contains extinction probabilities with different initial types. The extinction probabilities tend to $1$ as initial types tend to $-\infty$ and the extinction probabilities decreases monotonically as types tend to $\infty$. Notice that as type $i\to-\infty$, there are abrupt increments in extinction probabilities for passenger dominance cases. This phenomenon can be explained by negative selection imposing a ``barrier'' for downward drift of cell fitness.

    \begin{figure}[h]
        \centering
        \includegraphics[width=\textwidth]{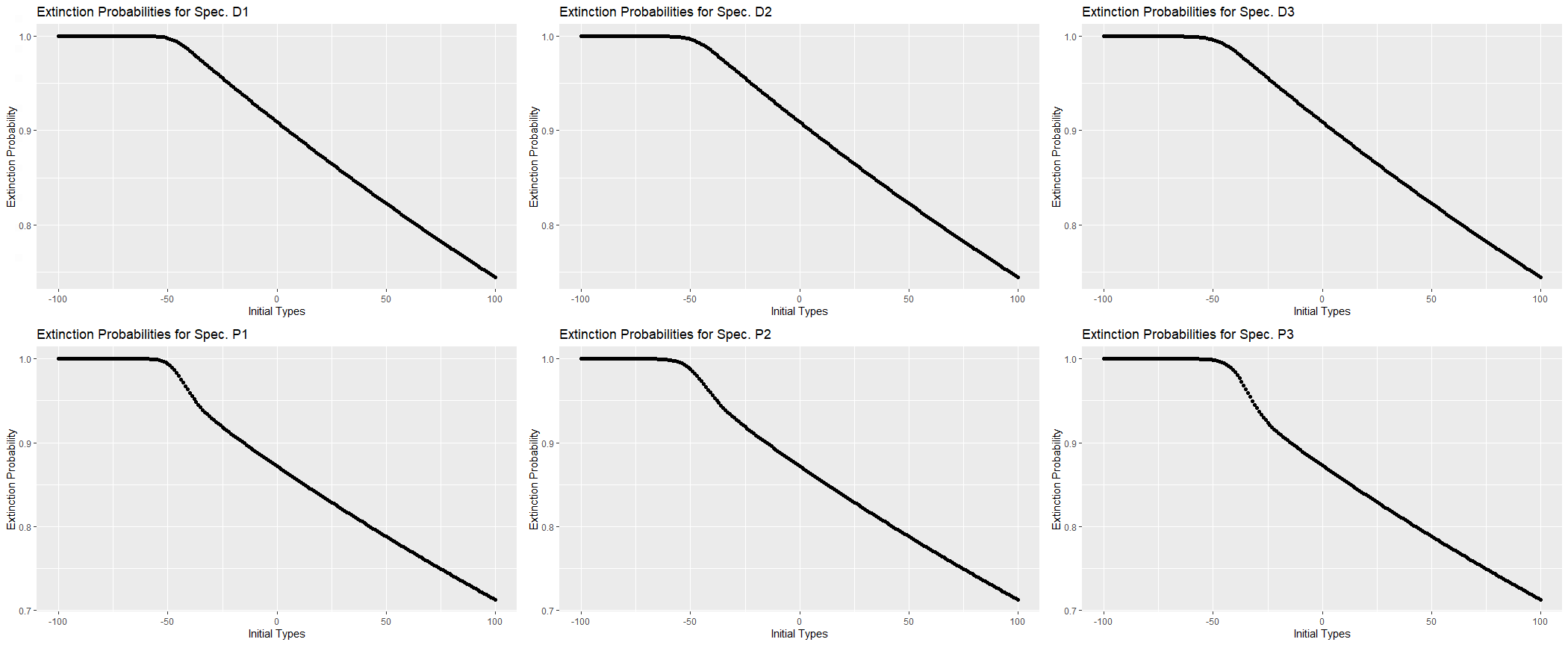}
        \caption{Computed extinction probabilities computed using the algorithm in \cite{hautphenne2013extinction} section $3.1$. The set of taboo types used to compute $q_{i}$ is $\mathbb{Z}\setminus\{i-100L,\cdots,i+100L\}$.}
        \label{extinction probabilities}
    \end{figure}

\subsection{Type Transition Process}\label{section 3.4}
	
Since we can only observe non-extinct populations, it makes sense to investigate the process conditional on non-extinction. In this section, we focus on a single non-extinct lineage and model the transition between types on the evolutionary path. This corresponds to investigation of transition between types in the backbone of the branching process. In \cite{durrett2015branching}, the backbone of a branching process consists of individuals that have descendants for infinitely many generations. This is feasible since we have shown extinction probabilities are less than $1$. Let $T_{i,j}$ denote the probability of a type $i$ cell transit to a type $j$ cell in the next generation conditional on non-extinction, then
	
\begin{align*}
   T_{i,i-1} = \frac{\nu}{\delta(i)}\frac{1-q_{i-1}}{1-q_{i}}; T_{i,i} = \frac{b(i)}{\delta(i)}\frac{1-q_{i}^{2}}{1-q_{i}}; T_{i,i+L} = \frac{\mu}{\delta(i)}\frac{1-q_{i+L}}{1-q_{i}}.
\end{align*}

\noindent We present a more detailed derivation of transition probabilities. Denote the initial population as $I$ and the backbone set associated with the initial population as $\mathcal{B}_{I}$. Define $ty(\cdot)$ to be a multiset-valued function that outputs the types of a collection of cells. For instance, $ty(I) = \{i,i,j\}$ if the $I$ contains two type $i$ cells and one type $j$ cell. Define $I'$ to be the set of first generation cells and $\mathbb{P}_{S}$ to be the probability measure such that $\mathbb{P}_{S}(ty(I)=S)=1$, where $S$ is a multiset. Since an individual's presence in the backbone implies all its ancestors’ {presence in the backbone, we have

\begin{align*}
    T_{i,i-1} &= \frac{\mathbb{P}_{\{i\}}(I\subset \mathcal{B}_{I}; I'\not\subset\mathcal{B}_{I}^{c}; ty(I') = \{i-1\})}{\mathbb{P}_{\{i\}}(I \subset \mathcal{B}_{I})} \\
    &= \frac{\mathbb{P}_{\{i\}}(I'\not\subset\mathcal{B}_{I}^{c}; ty(I') = \{i-1\})}{\mathbb{P}_{\{i\}}(I \subset \mathcal{B}_{I})} \\
    &= \frac{\mathbb{P}_{\{i\}}(ty(I') = \{i-1\})\mathbb{P}_{\{i-1\}}(I\subset \mathcal{B}_{I})}{\mathbb{P}_{\{i\}}(I \subset \mathcal{B}_{I})} \\
    &= \frac{\frac{\nu}{\delta(i)}(1-q_{i-1})}{1-q_{i}}.
\end{align*}
$T_{i,i+L}$ can be derived analogously. For $T_{i,i}$, we have
\begin{align*}
    T_{i,i} &= \frac{\mathbb{P}_{\{i\}}(I\subset \mathcal{B}_{I}; I'\not\subset\mathcal{B}_{I}^{c}; ty(I') = \{i,i\})}{\mathbb{P}_{\{i\}}(I \subset \mathcal{B}_{I})} \\
    &= \frac{\mathbb{P}_{\{i\}}(I'\not\subset\mathcal{B}_{I}^{c}; ty(I') = \{i,i\})}{\mathbb{P}_{i}(I \subset \mathcal{B}_{I})} \\
    &= \frac{\mathbb{P}_{\{i\}}(ty(I') = \{i,i\})\mathbb{P}_{\{i,i\}}(I \not\subset \mathcal{B}_{I}^{c})}{\mathbb{P}_{\{i\}}(I \subset \mathcal{B}_{I})} \\
    &= \frac{\frac{b(i)}{\delta(i)}(1-q_{i}^{2})}{1-q_{i}}.
\end{align*}

\noindent To show they sum to $1$, recall $\mathbf{G}(\mathbf{q}) = \mathbf{q}$.
	
\begin{align*}
    q_{i} &= \frac{d}{\delta(i)}+\frac{\nu}{\delta(i)}q_{i-1}+\frac{b(i)}{\delta(i)}q_{i}^{2}+\frac{\mu}{\delta(i)}q_{i+L} \\
    \Rightarrow 1 &= \frac{\nu}{\delta(i)}\frac{1-q_{i-1}}{1-q_{i}}+\frac{b(i)}{\delta(i)}\frac{1-q_{i}^{2}}{1-q_{i}}+\frac{\mu}{\delta(i)}\frac{1-q_{i+L}}{1-q_{i}}.
\end{align*}
	
To construct the type transition process on the continuous-time branching process, denoted $(X_{t})_{t\in\mathbb{R}_{+}}$, we construct its jump chain $(Y_{n})$ first. Since our model is initiated with a single type-$0$ cell, we have $X_{0} = Y_{0} = 0$. The transition probabilities for the jump chain $(Y_{n})$ are derived by normalizing $T_{i,i-1}$ and $T_{i,i+L}$. Therefore, the jump chain has transition probability matrix $J$ of the form
	
\begin{align*}
   J_{i,i-1} = \frac{\nu}{\delta(i)(1-T_{i,i})}\frac{1-q_{i-1}}{1-q_{i}};\;\; J_{i,i+L} = \frac{\mu}{\delta(i)(1-T_{i,i})}\frac{1-q_{i+L}}{1-q_{i}}.
\end{align*}

\noindent The distribution of the holding time for state $i$ can be represented as a random sum
	
\begin{align*}
    &\sum_{k=1}^{N_{i}}E_{i,k}, \text{ where } E_{i,k} \stackrel{IID}{\sim} \mathrm{Exp}(\delta(i)),\\
    \;\;&\text{where } N_{i} \sim Geom(1-T_{i,i}) \text{ on } \{1,2,\cdots\} \\
    &\text{and } \{N_{i}, E_{i,1},E_{i,2},\cdots\}\; \text{is an independent set of random variables}.
\end{align*}
	
\noindent The distribution of this random sum follows an $\mathrm{Exp}(\delta(i)(1-T_{i,i}))$ distribution. In Corollary of Lemma $\ref{lemma2}$, we show the supremum of $\delta(i)(1-T_{i,i})$ is bounded to conclude the process is non-explosive. 
	
Conditional increments for the jump chain of the type transition processes is

\begin{equation}
\begin{aligned}
    &\mathbb{E}(Y_{n+1}\mid Y_{n}=i) - i = \frac{\mu L(1-q_{i+L}) -\nu(1-q_{i-1})}{\delta(i)(1-q_{i})(1-T_{i,i})}.
\end{aligned}
\label{eq:cond_incre}
\end{equation}

\noindent By investigating the conditional increment of the jump chain, we arrive at a criterion to categorize them.
	
\subsubsection{Driver Dominance $(\nu\leq \mu L)$}

\begin{figure*}[h]
\centering
\includegraphics[width=\textwidth]{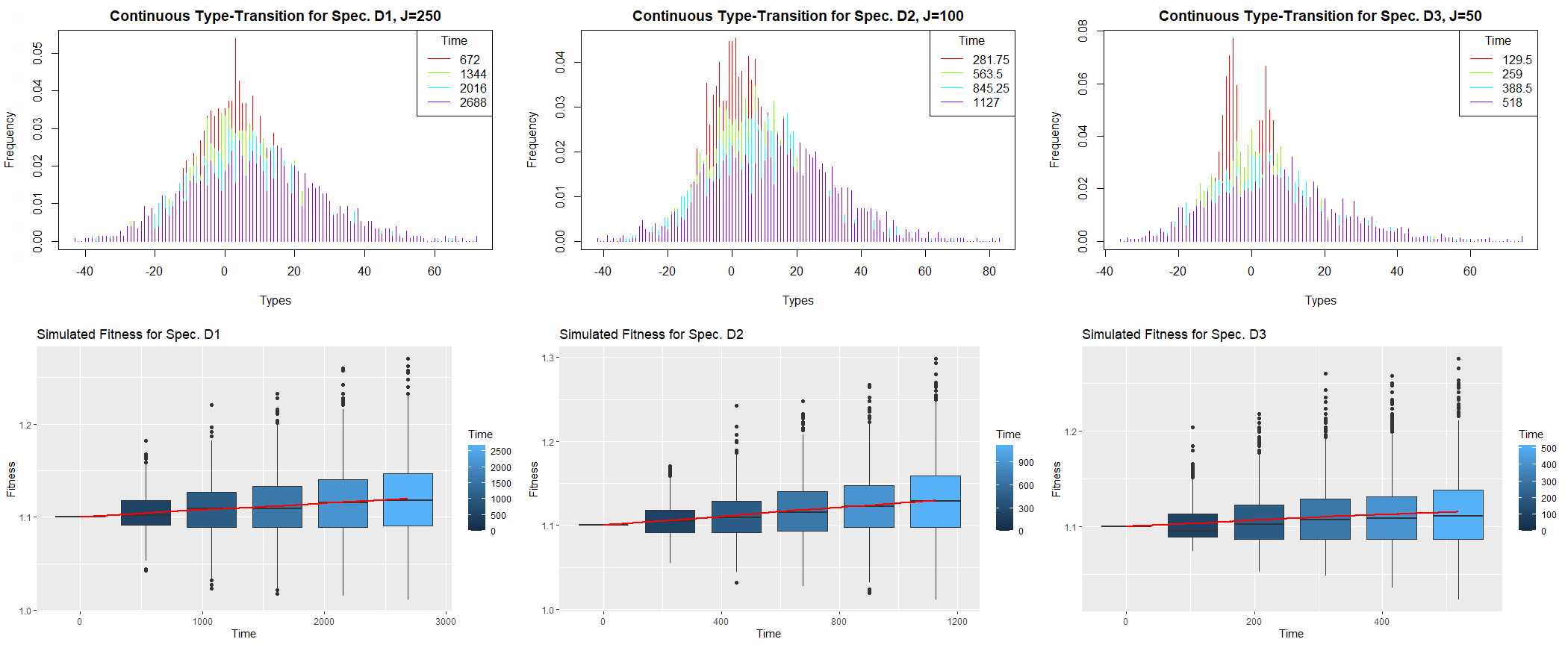}
\caption{Type transition process simulations with fitness boxplots and average fitness curves under driver dominance. We conduct $1500$ simulations for each parameter specification and $J$ is the number of jumps in each simulation. $J$ is chosen to be large enough to observe the limiting behavior of the type transition process. Simulations within a parameter specification have identical number of jumps and the final time displayed in the last row of legend is the minimal terminal time of all $1500$ simulations. The terminal times are truncated to avoid long decimal part.}\label{driver dominance}
\end{figure*}

Let us consider the numerator of the increment in \eqref{eq:cond_incre},
	
\begin{align*}
    &\mu L(1-q_{i+L}) -\nu(1-q_{i-1}) > 0, \forall i \in\mathbb{Z} \\
    \iff &\frac{\nu}{\mu L} < \frac{1-q_{i+L}}{1-q_{i-1}}, \forall i \in \mathbb{Z} \\
    \iff &\frac{\nu}{\mu L} \leq 1.
\end{align*}
	
\noindent Accordingly, both $(Y_{n})$ and $(X_{t})$ are submartingales when $\nu \leq \mu L$. Under driver dominance ($\nu \leq \mu L$), $\mathbb{E}(Y_{n})$ and $\mathbb{E}(b(Y_{n}))$ diverge to infinity due to to Lemma \ref{lemma3}. Analogously, $\mathbb{E}(X_{t})$ and $\mathbb{E}(b(X_{t}))$ both diverge to infinity as well. $\mathbb{E}(b(X_{t}))$ is the average fitness at time $t$.
	
According to Theorem \ref{thm2}, both $(Y_{n})$ and $(X_{t})$ are transient. A sufficient condition for transience given in the proof is $q_{i} \leq d/b(i)$ for all $i\in\mathbb{Z}$, which has a probabilistic interpretation. Since $\min\{d/b(i),1\}$ is the extinction probability of a cell population initiated by a type $i$ cell with $\mu = \nu = 0$, the condition suggests that under driver dominance, extinction is less likely to occur.

Due to transience, there are no limiting distributions for $(X_{t})$. According to Figure \ref{driver dominance}, the simulated type transition process becomes more and more diffuse in the positive direction as time unfolds. In addition, the simulated average fitness has a monotonic trend and diverges to $\infty$.
	
\subsubsection{Passenger Dominance $(\nu > \mu L)$}

\begin{figure*}[h]
\centering
\includegraphics[width=\textwidth]{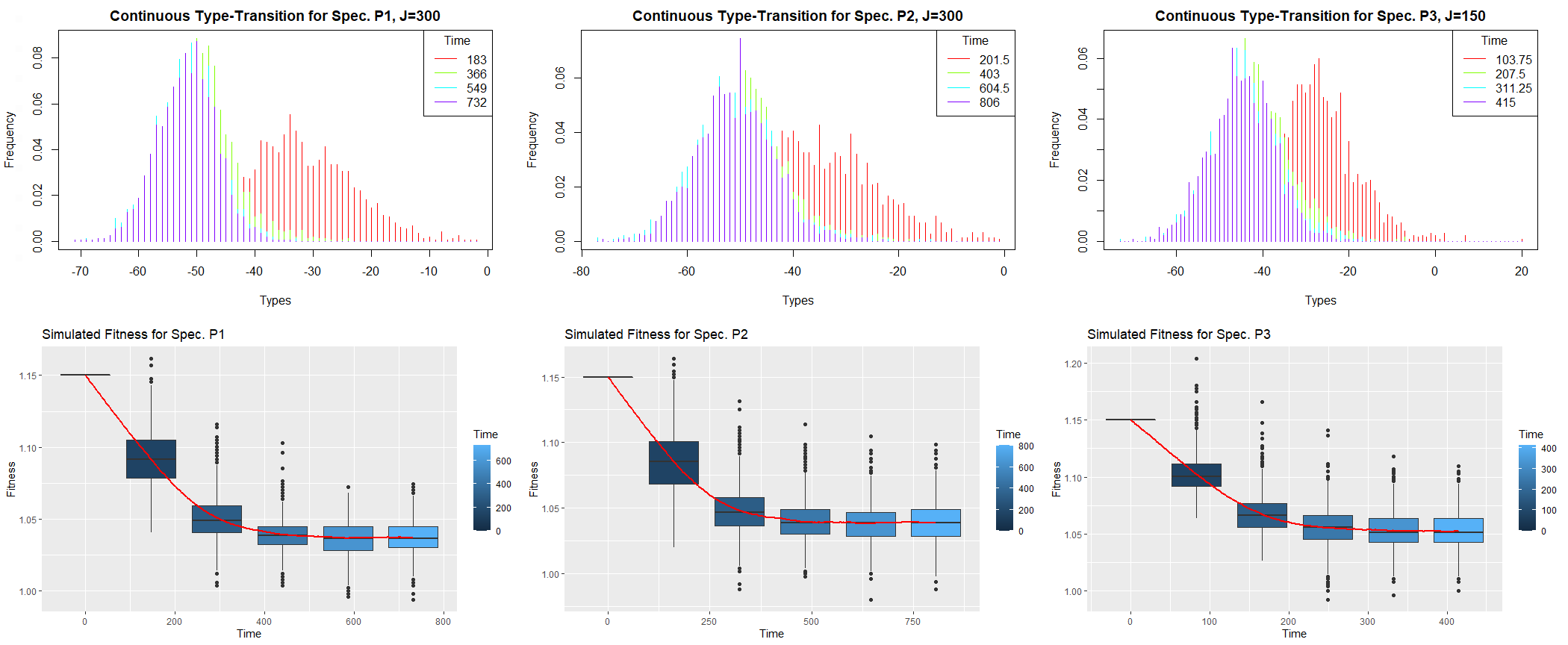}
\caption{Type transition process simulations with fitness boxplots and average fitness curve under passenger dominance. We conduct $1500$ simulations for each parameter specification and $J$ is the number of jumps in each simulation. $J$ is chosen to be large enough to observe the limiting behavior of the type transition process. Simulations within a parameter specification have identical number of jumps and the final time displayed in the last row of legend is the minimal terminal time of all $1500$ simulations. The terminal times are truncated to avoid long decimal part.}\label{passenger dominance}
\end{figure*}
	
Results in this section are based on Lemma \ref{lemma2}, which guarantees that the process cannot be a supermartingale if $\nu > \mu L$. Moreover, there exists an index $I$ such that $\mathbb{E}(Y_{n+1}\mid Y_{n}=i) > i$ for all $i< I$. Therefore, the average fitness will eventually stabilize, as shown in Figure \ref{passenger dominance}.

The index $I$ can be viewed as the location of the drift-barrier. Computed values for $I$ are correspondingly -53, -54, and -47, for specifications P1, P2, and P3 and they correspond to the abrupt changes in extinction probability in Figure \ref{extinction probabilities}. 
	
By Theorem \ref{thm3}, $(X_{t})$ admit a limiting distribution when $\nu > \mu L$. Simulations in Figure \ref{passenger dominance} show the distribution of the type transition process moves toward left and eventually stabilizes.

\newpage

\subsubsection{Implication of the Original Process}
Define $\lceil\log_{(1+s_{p})}(1+s_{d})\rceil = \Bar{L}$ and $\lfloor\log_{(1+s_{p})}(1+s_{d})\rfloor=\underline{L}$. We can approximate the original process from above or below by setting $L = \Bar{L}$ or $L = \underline{L}$. Since $\mu$ is small relative to $\nu$, $\mathbbm{1}_{\{\nu >\mu \Bar{L}\}}$ is likely to be the same as $\mathbbm{1}_{\{\nu >\mu \underline{L}\}}$. Therefore, for a wide range of mutation rates $\mu$ and $\nu$, driver/passenger dominance can be identified and existence of the limiting distribution for type transition can be inferred.  \\

\section{Discussion}
    
Cancer cells are believed to have elevated mutation rates. Mutations may be advantageous, neutral, or deleterious. In cancer, the relatively rare advantageous mutations are named ``drivers'', while the more frequent selectively neutral or slightly deleterious mutations are named ``passengers''. This convenient conceptual dichotomy prevails in biomedical and mathematical biology literature,  including \cite{mcfarland2014tug,mcfarland2017damaging}. Specifically, \cite{mcfarland2014tug} introduces the tug-of-war interaction between driver and passenger mutations which we analyze in the current paper, in the branching process framework. Since our model does not take environmental pressure in to account (fixed death rate), it is more suitable in modelling cancers with mild spatial constraints such as leukemia or other blood cancers. But what might be the threshold size? \cite{dang2003growth} document among other that 1 ccm of tumor tissue includes only $10^7 - 10^9$ tumor cells, which means that even quite large tumors may not experience severe constraints. Indeed, in biomathematical literature, this has been known for several decades, as illustrated by Figure $9$ in \cite{afenya1996acute} that suggests leukemic cells population maintains exponential growth for $700$ days to the order of $10^{12}$ cells. Returning to empirical estimates \cite{dinh2021predicting}, argue that acute myeloid leukemia cell population maintains exponential growth for around two years after therapy until reaching $10^{10} - 10^{11}$ cells. Our model might also be useful in modelling solid tumor progression if exponential growth is not violated. For instance, in \cite{ling2015extremely}, the site frequency spectrum of a large hepatocellular carcinoma tumor (slice of the tumor having 3.5 cm in diameter) is obtained from experimental data. The rigorous test of the observed site frequency spectrum of mutations in this specimen (sampled in ca. 300 sites) does not indicate a significant departure from the theoretical spectrum (Theorem $2$ in \cite{durrett2013population}) derived under the assumption of exponential growth.
    
In the present work, we remove part of the complexity of the original problem, by converting the two-dimensional type-space walk of each individual cell state $(n_d,n_p)\in\mathbb{Z}_+^2$ into a one-dimensional type-space random walk over $i\in\mathbb{Z}$ along the trajectories in which a new driver mutation corresponds to a jump up of $L$ steps, while a new passenger mutation corresponds to a jump down of $1$ step, where $L$ satisfies the relationship in \eqref{eq:dp_relation_man}. This leads us to a denumerable-type branching process of the type similar to the models in Chapter 8 of \cite{kimmel2015branching}, mathematically developed in \cite{hautphenne2013extinction}. Using results from these works, as well as those of \cite{seneta2006non}, we show that under driver dominance, the process is transient (``escapes to infinity''), while under passenger dominance, it is positive recurrent and leads to a limit distribution, given the infinitely long lineage.

In the context of population genetics, negative selection kicks in when passenger mutation rate $\nu$ exceeds a certain threshold $\mu L$ by imposing a barrier in the type transition process. The existence of this barrier guarantees the existence of a limiting distribution and stabilization of average fitness. This phenomenon, known as the ``drift barrier'' has been studied for a number of years in different population genetics models \cite{bachtrog2004adaptive, good2014deleterious, good2012distribution, rouzine2003solitary}. The difference between the classical drift barrier and the effects we observe is that we see stability at the bottom end of the state ladder, while the classical authors observe a reverse effect that slows down the upward movement towards perfect adaptation.

Under downward-trend regime (passenger dominance) there is with positive probability a ``reservoir'' of cells with a wide range of types (fitnesses). In the context of cancer a cell population, this may mean that ``indolent'' cancer cell colonies may allow the biological process to rebound if conditions change, as in the ``punctuated equilibria'' theory of cancer evolution in \cite{gao2016punctuated} and \cite{davis2017tumor}.
 
The current model is an alternative version of the models used by \cite{kurpas2022modes} to obtain, computationally, the site frequency spectra (SFS) of a range of types of human tumors. However, those models were framed in the terms of a Moran process with constant population size. This other setting allows us to model a form of environmental constraints, as the competitive malignant clones are gradually suppressing normal tissue (see Figure 15 in \cite{kurpas2022modes}); however differently from the saturation mechanism assumed in \cite{mcfarland2014tug}. The Moran-process Tug-of-War was also used by \cite{kurpas2022genomic} to model evolution of consecutive waves of viral infections. In that setting, the individuals are not cells but viral genomes. Consecutive waves of virus variants modeled by the Tug-of-War, depicted in Figure 15 in \cite{kurpas2022genomic}, are qualitatively similar to the data-based version in their Figure 7. Recently, a paper was published by \cite{bobrowski2023moran} which considers the Tug-of-War process, using a different approach based on operator semigroup theory. The results of that paper do not overlap with our current work.

\section{Acknowledgements}
We thank Professor Peter Braunsteins and Professor Sophie Hautphenne for suggesting the coupling proof in Theorem 1. Ren-Yi Wang was supported by the Department of Statistics at Rice University and by the NIH grant R01 268380 to Wenyi Wang. Marek Kimmel was supported in part by the NSF DMS RAPID COLLABORATION grant 2030577 to  Marek Kimmel and Simon Tavar\'e, and by the National Science Center (Poland) grant number 2021/41/B/NZ2/04134.

\clearpage

\appendix
\section{Model Construction}\label{model_construction}
\textbf{Assumptions: } $s_{d}>s_{p}$, $\mu < \nu < d$. 

The type of the branching process is determined by the number of driver and passenger mutations $(n_{d},n_{p})$. A type $(n_{d},n_{p})$ cell has fitness $b_{0}(1+s_{d})^{n_{d}}/(1+s_{p})^{n_{p}}$. We construct the model under the following driver-passenger relation.

\begin{equation}
\begin{aligned}
    \log_{(1+s_{p})}(1+s_{d})=L \in \mathbb{N}, L \geq 2.
\end{aligned} 
\label{eq:dp_relation}
\end{equation}

\noindent Therefore, the fitness function can be rewritten as $b(i) = b_{0}(1+s_{p})^{i}$. Mutation rates for driver and passenger mutations are $\mu$ and $\nu$, respectively. The death rate is $d$  and we define the total rate to be $\delta(i) = b(i)+\mu+\nu+d$.

Denote the continuous process as $(\mathbf{Z}_{t})_{t\in\mathbb{R}_{+}}$ and the jump chain as $(\mathbf{J}_{n})_{n\in\mathbb{Z}_{+}}$. In this construction, we assume that the parameters of the fitness function of drivers and passengers satisfy condition \eqref{eq:dp_relation} for the reduced process. Therefore, driver mutation pushes the type of the cell $L$-units up, while a passenger mutation pulls it $1$-unit down. 
	
A na\"ive state space might be $\mathbb{S} \simeq \mathbb{Z_{+}}^{\mathbb{Z}}\simeq \mathbb{R}$ which is uncountable. However, we can define the state space in a more parsimonious way. We set
$S_{m}=\{\cdots,0,n_{-m},\cdots,n_{m},0,\cdots \mid n_{k}\in \mathbb{Z}_{+}; -m\leq k\leq m\}$ and let $\mathcal{S}_{\infty}=\cup_{m\geq 0}S_{m}$. Since a countable union of countable sets is countable, the set $S_{\infty}$ is countable. Take an element in $\mathbf{s} \in \mathbb{S}\setminus S_{\infty}$, then it has to have nonzero elements with indices arbitrarily small or arbitrarily large. Due to the proliferation mechanism and initial condition $\mathbf{J}_{0}=\underline{e}_{0}$, $\mathbb{P}(\mathbf{J}_{n} =\mathbf{s}) = 0$ for each $n\geq 0$. As long as the initial population is of finite size, we may discard the states with zero probabilities and the new state space $S_{\infty}$ is countable. 
	
To construct the continuous-time process $\mathbf{Z}_{t}$, we associate a holding time following a $Exp(\lambda_{\mathbf{s}})$ distribution to state $\mathbf{s} = (\cdots,n_{-1},n_{0},n_{1},\cdots)$ where $\lambda_{\mathbf{s}} = \sum_{k\in\mathbb{Z}}n_{k}\delta(k)$. Let $j(n)$ be the time of the $n$th jump. The proliferation mechanism is as follows. After the $n$th jump, there are $(\mathbf{Z}_{j(n)})_{k}$ individuals of type $k$; then a single type $i$ cell will be selected to proliferate, mutate, or die with probability 
	
\begin{align*}
    \frac{(\mathbf{Z}_{j(n)})_{i}\delta(i)}{\sum_{k\in\mathbb{Z}}(\mathbf{Z}_{j(n)})_{k}\delta(k)}.
\end{align*}
	
\noindent The probability is due to the competing exponential random variables argument. The cell will then proliferate, acquire a driver mutation, acquire a passenger mutation, or die with probabilities 

\begin{align*}
    \frac{b(i)}{\delta(i)}, \frac{\mu}{\delta(i)},\frac{\nu}{\delta(i)} \text{, or } \frac{d}{\delta(i)} \text{, respectively.}
\end{align*}
	
Finally, the probability space for the jump chain can be constructed according to \cite{williams1991probability} page $105$ and the space can be extended to support the continuous process. In this space, the probabilities for the embedded branching process $(\mathbf{E}_{n})$ can be extracted.

The probability generating function for the embedded branching process is $\mathbf{G}$ and its coordinates are defined by 

\begin{align*}
    G_{i}(\mathbf{s}) = \frac{d}{\delta(i)} + \frac{\nu}{\delta(i)}s_{i-1}+ \frac{b(i)}{\delta(i)}s_{i}^{2}+ \frac{\mu}{\delta(i)}s_{i+L}.
\end{align*}

\noindent The mean matrix is 

\begin{align*}
    &M_{i,i-1} = \frac{\nu}{\delta(i)}, M_{i,i} = 2\frac{b(i)}{\delta(i)}, M_{i,i+L} = \frac{\mu}{\delta(i)}; \forall i\in\mathbb{Z}.
\end{align*}


\section{Propositions}
\begin{prop}\label{prop1}
$\forall i,j \in \mathbb{Z},\lim_{n\to\infty}[(M^{n})_{i,j}]^{1/n}=2$, which implies $\mathbf{q}\leq \tilde{\mathbf{q}}< \underline{1}$.
\end{prop}
    
\begin{proof}
Recall that the mean matrix $M$ is a matrix whose nonzero entries are

\begin{align*}
    &M_{i,i-1} = \frac{\nu}{\delta(i)}, M_{i,i} = 2\frac{b(i)}{\delta(i)}, M_{i,i+L} = \frac{\mu}{\delta(i)}; \forall i\in\mathbb{Z}.
\end{align*}

\noindent The proof uses notation and results from \cite{seneta2006non}. Since all matrices there are indexed by $i,j \in \mathbb{N}$, we relabel our mean matrix $M$ according to the relabeling function $\pi$ and denote it $\widetilde{M}$. The relabeling function $\pi$ is a mapping from $\mathbb{Z}$ to $\mathbb{Z}_{+}$ such that

\begin{equation}
\begin{aligned}
    \pi(0) &= 0,\pi(1)=1,\cdots, \pi(L) = L, \\
    \pi(-1) &= L+1,\pi(-2) = L+2,\cdots, \pi(-L) = 2L, \\
    \pi(L+1) &= 2L+1,\pi(L+2) = 2L+2,\cdots, \pi(2L) = 3L, \\
    &\quad\vdots
\end{aligned}
\label{eq:relabel}
\end{equation}

\noindent The $n$th truncation (the northwestern $n\times n$ submatrix) of $\widetilde{M}$ is irreducible for $n \geq L+1$. We will denote the $n$th truncation of $\widetilde{M}$ by $_{(n)}\widetilde{M}$. Define $R$ to be the convergence parameter of $\widetilde{M}$ and $r = 1/R$. An explicit calculation of $r$ does not seem feasible; however, the fact that $M_{i,i} \to 2$ as $i \to \infty$ simplifies the matter. For fixed $x\in(1,2)$, there exists $n$  such that $M_{n,n} > x$. Since the matrix is non-negative,  
    
\begin{align*}
    &(M^{k})_{n,n} \geq (M_{n,n})^{k} \geq x^{k} \\
    \Rightarrow &\lim_{k\to\infty}[(M^{k})_{n,n}]^{1/k} \geq x, \forall 1 < x < 2\\
    \Rightarrow &r\geq 2, R \leq \frac{1}{2}. 
\end{align*}
    
\noindent If the convergence parameter of the $n$th truncation of $\widetilde{M}$ is $_{(n)}R$ where $n \geq L+1$, then $1/{_{(n)}R}$ is the Perron-Frobenius eigenvalue of $_{(n)}\widetilde{M}$. By the Perron-Frobenius theorem, we have 
    
\begin{align*}
    \frac{1}{_{(n)}R} \leq \max_{i}\sum_{j}\;_{(n)}\widetilde{M}_{i,j} \leq 2.
\end{align*}
    
\noindent By Theorem $6.8$ in \cite{seneta2006non}, $_{(n)}R \to R$ as $n\to\infty$. Therefore, $\frac{1}{R}\leq 2$ and this implies $R = \frac{1}{2}$ or equivalently, $r = 2$.

Since the $k$th truncation of $\widetilde{M}$ is irreducible for $k\geq L+1$ and $r>1$, we invoke Proposition $4.1$ in \cite{hautphenne2013extinction} and conclude $\tilde{\mathbf{q}}<\underline{1}$, where $\underline{1}$ is the vector $(\cdots,1,1,1,\cdots)^{T}$.
\end{proof}

\noindent \textbf{Remark:} Although $\widetilde{M}$ is irreducible, it does not guarantee the PGF has at most $2$ fixed points. In \cite{bertacchi2014strong}, authors show that in the case of non-strong local survival, it is possible to have $\mathbf{q}<\tilde{\mathbf{q}}<\underline{1}$ with an irreducible mean matrix. Corollary $4.4$ of \cite{bertacchi2020extinction} states that $\sup_{i}\tilde{q}_{i}<1$ is sufficient for $\mathbf{q}=\tilde{\mathbf{q}}$. However, this condition is not satisfied in our model and we cannot conclude $\mathbf{q}=\tilde{\mathbf{q}}$.


\section{Lemmas}
\begin{lemma}\label{lemma1}
    Under the assumption $(d+\nu) > \mu L$, the continuous-time process $\mathbf{Z}_{t}$ is non-explosive.
\end{lemma}
	
\begin{proof}
We compare the process $\mathbf{Z}_{t}$ to a simpler process that is more likely to explode. Consider a new passenger mutations rate that incorporates the death rate $\xi = d+\nu$ and set the death rate of this new process to $0$. For the new process to be non-explosive, it suffices to show 
	
\begin{align*}
    \sum_{n=0}^{\infty} \frac{1}{\lambda_{\mathbf{J}_{n}}} = \infty \text{ with probability } 1,
\end{align*}
	
\noindent where $(\mathbf{J}_{n})$ is the jump chain of the modified process with passenger mutation rate $\xi$. If the new process were a pure birth process, the above sum would be a multiple of the standard harmonic series and diverge to $\infty$ with probability $1$. Under the assumption $\xi > \mu L$, a lineage is more likely to acquire $L$ passenger mutations than acquire one driver mutation. To see this, let us focus on the transition between types along a lineage. We have
	
\begin{align*}
    &\mathbb{P}(\text{Consecutive } L \text{ passenger mutations}) = (\frac{\xi}{\mu+\xi})^{L}\geq (\frac{L}{L+1})^{L} > \frac{1}{L+1} \\
    &\mathbb{P}(1 \text{ driver mutation}) = \frac{\mu}{\mu+\xi} < \frac{1}{L+1}.
\end{align*}
	
Intuitively, this random sum of holding times should diverge to infinity almost surely since the fitness of cells on a lineage tends to decrease along generations. For a rigorous argument, since there are at most countably many lineages and we may enumerate them as $l_{1},l_{2},\cdots$. If the supremum of fitnesses of all lineages across generation is bounded by $M$ from above, we have 
	
\begin{align*}
    \sum_{n=0}^{\infty} \frac{1}{\lambda_{\mathbf{J}_{n}}} \geq \frac{1}{M+\mu+\xi} \sum_{n=0}^{\infty} \frac{1}{N_{n}} \geq \frac{1}{M+\mu+\xi} \sum_{k=1}^{\infty} \frac{1}{k}  = \infty,
\end{align*}
	
\noindent where $N_{n}$ is the population size after the $n$th jump and it is non-decreasing with $n$. Hence, for an explosion to occur, we must have at least one lineage whose fitness is unbounded from above. We now show this event has probability zero. Select a lineage and model the transition of cell type as random walk $(R_{n})$ that increases by $L$ with probability $\frac{\mu}{\mu+\xi}< \frac{1}{L+1}$ and decreases by $1$ with probability $\frac{\xi}{\mu+\xi}> \frac{L}{L+1}$. According to Kolmogorov's zero-one law, $\limsup_{n\to\infty} R_{n}$ and $\liminf_{n\to\infty} R_{n}$ are almost surely constants (see page $88$ in \cite{cinlar2011probability}). By Markov's inequality, for some $\alpha > 0$ that will be specified later, we have 

\begin{equation}
\begin{aligned}
    \mathbb{P}(R_{n} \geq K) = \mathbb{P}(e^{\alpha R_{n}} \geq e^{\alpha K}) \leq \frac{(\frac{\xi}{\mu+\xi}e^{-\alpha}+\frac{\mu}{\mu+\xi}e^{L\alpha})^{n}}{e^{K\alpha}}.
\end{aligned}
\label{eq:cond}
\end{equation}
	
\noindent If there exists $\alpha > 0$ such that $\frac{\xi}{\mu+\xi}e^{-\alpha}+\frac{\mu}{\mu+\xi}e^{L\alpha}< 1$, $\sum_{n=0}^{\infty}\mathbb{P}(R_{n} \geq K) < \infty$. Let $x = e^{\alpha}$, then the desired condition is equivalent to the existence of some $x > 1$ such that
	
\begin{equation}
\begin{aligned}
    \frac{\xi}{\mu+\xi}\frac{1}{x}+\frac{\mu}{\mu+\xi}x^{L}< 1 \iff \frac{\xi}{\mu+\xi}+\frac{\mu}{\mu+\xi}x^{L+1}- x<0.
\end{aligned} 
\label{eq:equiv_cond}
\end{equation}
	
\noindent Observe that the RHS of $\eqref{eq:equiv_cond}$ is equal to $0$ when $x=1$. The derivative of $\frac{\xi}{\mu+\xi}+\frac{\mu}{\mu+\xi}x^{L+1}- x$ is $\frac{\mu}{\mu+\xi}(L+1)x^{L}- 1$, and it is negative for $x$ slightly greater than $1$ since $\frac{\mu}{\mu+\xi} < \frac{1}{L+1}$. Therefore, there exists $\alpha>0$ such that $\frac{\xi}{\mu+\xi}e^{-\alpha}+\frac{\mu}{\mu+\xi}e^{L\alpha}< 1$. The inequality in $\eqref{eq:cond}$ now proves $\mathbb{P}(R_{n}\geq K)$ is summable with respect to $n$. By Borel-Cantelli lemma, we have $\mathbf{1}_{\{R_{n}\geq K\}} \to 0$ almost surely and this implies $\limsup_{n\to\infty}R_{n} < K$ for all $K$. Hence, $\lim_{n\to\infty} R_{n} = -\infty$ and finally
	
\begin{align*}
    \mathbb{P}(\text{At least one lineage has unbounded fitness}) &\leq \sum_{i=1}^{\infty}\mathbb{P}(l_{i} \text{ has unbounded fitness})=0.
\end{align*}

\end{proof}

\noindent\textbf{Corollary:} If $\nu + d > \mu L$, the continuous process $(\mathbf{Z}_{t})$ becomes extinct if and only if the embedded process $(\mathbf{E}_{n})$ becomes extinct.

\begin{proof}
To show the equivalence of extinctions in the continuous and embedded process under non-explosion, let $A_{n}$ denote the event that the embedded process becomes extinct at or before the $n$th generation. That is, $\omega \in A_{n}$ implies $||\mathbf{E}_{n}(\omega)||_{\ell^{1}} = 0$. Since the number of generations $n$ is finite, the number of jumps is also finite. Therefore, $\lim_{t\to\infty}||\mathbf{Z}_{t}(\omega)||_{\ell^{1}} = 0$ for almost every $\omega \in A_{n}$. 

On the other hand, let $B_{t}$ be the event that the continuous-time process becomes extinct at or before time $t \in \mathbb{Q}_{+}=\{q\in\mathbb{Q}\mid q \geq 0\}$. That is, $\omega \in B_{n}$ implies $||\mathbf{Z}_{t}(\omega)||_{\ell^{1}} = 0$. Since $[0,t]$ is a finite interval and the process is non-explosive, the number of jumps in $[0,t]$ is a.s. finite. Hence, the number of generations in $[0,t]$ is also a.s. finite and $\lim_{n\to\infty}||\mathbf{E}_{n}(\omega)||_{\ell^{1}} = 0$ for almost every $\omega \in B_{t}$. 

Finally, let $A = \cup_{n\geq 0}A_{n}$ and $B=\cup_{t\in\mathbb{Q}_{+}}B_{t}$, then 

\begin{align*}
    \mathbb{P}(A) = \mathbb{P}(\cup_{n> 0}A_{n}) =  \mathbb{P}(\cup_{t\in\mathbb{Q}_{+}}B_{t}) = \mathbb{P}(B).
\end{align*}
\end{proof}


\begin{lemma}\label{lemma2}
    Let $g(x) = \mu x^{L+1}-(\mu+\nu+d)x+\nu = 0$; then under condition of non-explosion ($\nu+d > \mu L$) and $d>\nu$, $\lim_{i\to-\infty}\frac{1-q_{i}}{1-q_{i-1}} = \alpha$ where $\alpha$ is the unique real solution to $g(x)=0$ such that $x \geq 1$. As a consequence, $\frac{\nu}{\mu L} < \alpha^{L+1}$.
\end{lemma}
    
\begin{proof}
    The difference equation defining extinction probabilities can be transformed into a perturbed linear system by the following manipulation. Recall that extinction probabilities satisfy
    
    \begin{equation}
    \begin{aligned}
        1-s_{i-1} = \frac{\mu+\nu+d}{\nu}(1-s_{i}) - \frac{\mu}{\nu}(1-s_{i+L}) -\frac{b(i)}{\nu}s_{i}(1-s_{i}), \forall i\in\mathbb{Z}.
    \end{aligned}
    \label{eq:ext_prob_diff_eq}
    \end{equation}
    
    \noindent Define $y_{n} = 1-s_{-n}$, then we obtain the following difference equation.
    
    \begin{equation}
    \begin{aligned}
        y_{n+1} = \left(\frac{\mu+\nu+d}{\nu} - \frac{b(-n)}{\nu}(1-y_{n})\right) y_{n} - \frac{\mu}{\nu}y_{n-L}, \forall n\in\mathbb{Z}.
    \end{aligned}
    \label{eq:ext_prob_diff_eq_y}
    \end{equation}
    
    \noindent Since $y_{n} = 1-q_{-n}$ solves $\eqref{eq:ext_prob_diff_eq_y}$, replacing $1-y_{n}$ by $q_{-n}$ will result in a new system whose solution set contains $(1-q_{-n})_{n=-\infty}^{\infty}$. we obtain a perturbed linear system $\eqref{eq:pert_sys}$ with $(1-q_{-n})_{n\in\mathbb{Z}}$ being one of its solutions. 
    
    \begin{equation}
    \begin{aligned}
        y_{n+1} = \left(\frac{\mu+\nu+d}{\nu} - \frac{b(-n)}{\nu}q_{-n}\right)y_{n} - \frac{\mu}{\nu}y_{n-L}, \forall n\in\mathbb{Z}.
    \end{aligned}
    \label{eq:pert_sys}
    \end{equation}
    
    \noindent Denote $\mathbf{y}_{n} = (y_{n},\cdots,y_{n-L})^{T}$. Then we can rewrite the above relation as 
    
    \begin{align*}
        &\mathbf{y}_{n+1} = (A+R(n))\mathbf{y}_{n}, \\
        &A =
        \begin{pmatrix}
            \frac{\mu+\nu+d}{\nu} &\underline{0}^{T} &-\frac{\mu}{\nu} \\
            1 &\underline{0}^{T} &0 \\
            \underline{0} &I_{L-1} &\underline{0}
        \end{pmatrix},
        \;\;R(n) = \text{diag}(-\frac{b(-n)}{\nu}q_{-n}, \cdots, -\frac{b(-n+L)}{\nu}q_{-n+L}),
    \end{align*}
    
    \noindent where $\underline{0}$ is a $(L-1)-$column vector and $I_{L-1}$ is the $(L-1)\times(L-1)$ identity matrix. The characteristic polynomial of matrix $A$ is a multiple of $f(x) = \nu x^{L+1} - (\mu+\nu+d)x^{L} + \mu$. Since $x=0$ is no a root of $f$, if we define $g(x) = x^{L+1}f(1/x) = \mu x^{L+1} - (\mu+\nu+d)x +\nu$ for $x\in\mathbb{R}$, there is a one-to-one correspondence between roots of $f$ and roots of $g$. That is, $f(x)=0$ if and only if $g(1/x)=0$. Observe that $g(1) = -d < 0$ and $g''(x) > 0$ on $(0,\infty)$. Therefore, there is a unique positive real root of modulus greater than $1$ for $g$ and we denote it as $\alpha$. Note that $\frac{1}{\alpha}$ is the unique real root of $f$ that has modulus less than $1$. To show that roots of $g$ are simple, consider the condition
        
    \begin{align*}
        g'(x) = 0 \Rightarrow x^{L} = \frac{\mu+\nu+d}{(L+1)\mu}.
    \end{align*}
        
    \noindent Substituting $x^{L}$ in $g$, we obtain 
        
    \begin{align*}
        x\frac{\mu+\nu+d}{L+1} - (\mu+\nu+d)x +\nu = 0 \Rightarrow x = \frac{L+1}{L}\frac{\nu}{\mu+\nu+d}.
    \end{align*}
    
    \noindent Hence, $x^{L+1} = \frac{\nu}{\mu L}$ and if $x = (\frac{\nu}{\mu L})^{\frac{1}{L+1}}$ is not a root of $g$, all roots of $g$ are simple and matrix $A$ is diagonalizable. For contradiction, suppose $g((\frac{\nu}{\mu L})^{\frac{1}{L+1}}) = 0$, which implies $(\frac{\nu}{\mu L})^{\frac{1}{L+1}} = \frac{L+1}{L}\frac{\nu}{\mu+\nu+d}$. Under the non-explosion condition $\nu+d > \mu L$, this yields
    
    \begin{align*}
        \mu+\nu+d > \mu(L+1) \Rightarrow\; &\frac{\nu}{\mu+\nu+d} < \frac{\nu}{\mu (L+1)}  \\
        \Rightarrow\; &(\frac{\nu}{\mu+\nu+d}\frac{L+1}{L})^{\frac{1}{L+1}} < (\frac{\nu}{\mu L})^{\frac{1}{L+1}}=\frac{L+1}{L}\frac{\nu}{\mu+\nu+d}  \\
        \Rightarrow\; &1 < \frac{L+1}{L}\frac{\nu}{\mu+\nu+d} \\
        \Rightarrow\; &L\mu + Ld < \nu, \text{ which contradicts } d > \nu.
    \end{align*}
        
    We now express $A$ as $A = T\Lambda T^{-1}$ with $\Lambda$ being a diagonal matrix containing distinct eigenvalues of $A$. Let us choose $N \in\mathbb{N}$, and investigate the asymptotic behavior of the perturbed linear system 
    
    \begin{align*}
        \mathbf{z}_{n+1} = T\mathbf{y}_{n+1} = T(A+R(n))T^{-1}T\mathbf{y}_{n} = (\Lambda +TR(n)T^{-1})\mathbf{z}_{n}, n\geq N.
    \end{align*}
    
    \noindent Since all eigenvalues of $A$ are nonzero and simple and the operator norm of the perturbation $||TR(n)T^{-1}||$ satisfies

    \begin{align*}
        &||TR(n)T^{-1}|| \leq ||T||\cdot||R(n)||\cdot||T^{-1}|| \leq \frac{b(-n+L)}{\nu}||T||\cdot||T^{-1}|| \\
        \Rightarrow &\sum_{n\geq 0}||TR(n)T^{-1}|| < \infty.
    \end{align*}

    \noindent Hence, can invoke Theorem $3.4$ in \cite{bodine2015asymptotic} (Discrete Version of Levinson’s Fundamental Theorem) with $n_{0}=0,K_{1}=K_{2}=1$. For convenience, the statement of the theorem is provided in Section $\ref{Levinson}$. The fundamental matrix for $\mathbf{z}_{n}$ has the following form
    
    \begin{align*}
        (I + o(1)) \Lambda^{n} \text{ as } n \to \infty.
    \end{align*}
    
    \noindent Equivalently, the fundamental matrix for $\mathbf{y}_{n}$ has the following form
        
    \begin{align*}
        &(I + o(1)) T^{-1}\Lambda^{n} \text{ as } n \to \infty, \text{ where }\\
        &T^{-1} = 
        \begin{pmatrix}
            \lambda_{1}^{L} &\cdots &\lambda_{L+1}^{L} \\
            \lambda_{1}^{L-1} &\cdots &\lambda_{L+1}^{L-1} \\
            \vdots & &\vdots \\
            1 &\cdots &1 \\
        \end{pmatrix}.
    \end{align*}
        
    \noindent Since $(y_{n}) = (1-q_{-n})$ solves the perturbed system,

    \begin{align*}
        1-q_{-n} &= (1+o(1))\sum_{k=1}^{L+1}c_{k}\lambda_{k}^{n}.
    \end{align*}
    
    The dominating eigenvalue is real. To prove this, suppose the dominating eigenvalue is a complex number $r\exp(\iota \theta)$ with $\theta \notin \{0,\pi\}$, then $r\exp(-\iota \theta)$ is also a dominating eigenvalue. Hence, as $n\to\infty$ and omit all $n$ such that $\cos(n\theta) = 0$,

    \begin{align*}
        \frac{1-q_{-n}}{r^{n}} \sim C \cdot \cos(n\theta), C \neq 0 \text{ is a constant}.
    \end{align*}
    
    \noindent This is impossible since LHS is always positive and the RHS takes negative values infinitely many times as $n\to\infty$. Similarly, we cannot have $2n$ dominating complex eigenvalues. Hence,
    
    \begin{align*}
        \lim_{n\to\infty}\frac{1-q_{-n-1}}{1-q_{-n}} = \lambda_{k}\in\mathbb{R} \text{ for some } 1\leq k\leq L+1.
    \end{align*}

    \noindent Recall that survival probabilities satisfy 
    
    \begin{align*}
        &1-q_{-n} = \frac{\nu}{\delta(-n)}(1-q_{-n-1}) + \frac{b(-n)}{\delta(-n)}(1-q_{-n}^{2})+\frac{\mu}{\delta(-n)}(1-q_{-n+L}) \\
        \Rightarrow &\frac{1-q_{-n}}{1-q_{-n-1}}=  \frac{\nu}{\delta(-n)} + \frac{b(-n)}{\delta(-n)}\frac{1-q_{-n}^{2}}{(1-q_{-n-1})}+\frac{\mu}{\delta(-n)}\frac{1-q_{-n+L}}{1-q_{-n-1}} \\
        \Rightarrow &\frac{1}{\lambda_{k}} = \frac{\nu}{\mu+\nu+d} + \frac{\mu}{\mu+\nu+d}(\frac{1}{\lambda_{k}})^{L+1}.
    \end{align*}
        
    \noindent By Theorem $\ref{thm1}$, monotonicity of extinction probabilities implies that $\lambda_{k}< 1$. Since $\lambda_{k}$ is the limit of ratios of real numbers, it must be a real number as well. Therefore, $\lambda_{k} = \frac{1}{\alpha}$ and $\alpha$ satisfies $g(\alpha) = 0$. Finally, suppose $\frac{\nu}{\mu L} \geq \alpha^{L+1}$, we arrive at a contradiction in the following way,
	
    \begin{align*}
        \frac{\nu}{\mu L} \geq \alpha^{L+1} = \frac{(\mu+\nu+d)\alpha -\nu}{\mu} \Rightarrow \alpha \leq \frac{L+1}{L}\frac{\nu}{\mu+\nu+d} < \frac{2}{1}\frac{\nu}{\nu+d} <1.
    \end{align*}
\end{proof}

\noindent \textbf{Corollary: } The type transition process $(X_{t})$ is non-explosive.

\begin{proof}
    According to Theorem 2.7.1 in \cite{norris1998markov}, it suffices to show the supremum of rates is bounded. 

    \begin{align*}
    &\begin{cases}
        \lim_{i\to\infty}\delta(i)(1-T_{i,i}) = \lim_{i\to\infty}\mu\frac{1-q_{i+L}}{1-q_{i}}+\nu\frac{1-q_{i-1}}{1-q_{i}} = \mu+\nu \\
        \lim_{i\to-\infty}\delta(i)(1-T_{i,i}) = \mu\alpha^{L}+\nu\alpha^{-1} = \mu+\nu+d
    \end{cases} \\
    \Rightarrow &\sup_{i\in\mathbb{Z}}\{\delta(i)(1-T_{i,i})\} <\infty.
    \end{align*}
\end{proof}

\begin{lemma}\label{lemma3}
    Under driver dominance $(\nu/\mu L\leq  1)$, $\lim_{n\to\infty}\mathbb{E}(Y_{n}) = \infty$ and $\lim_{n\to\infty}\mathbb{E}(b(Y_{n}))\to\infty$. Similarly, $\lim_{t\to\infty}\mathbb{E}(X_{t}) = \infty$ and $\lim_{t\to\infty}\mathbb{E}(b(X_{t}))\to\infty$ due to non-explosion.
\end{lemma}
    
\begin{proof}
    Notice that 

    \begin{align*}
        \frac{J_{i,i+L}}{J_{i,i-1}} = \frac{\mu(1-q_{i+L})}{\nu(1-q_{i-1})} \geq \frac{1}{L}\frac{1-q_{i+L}}{1-q_{i-1}} > \frac{1}{L}.
    \end{align*}

    \noindent Therefore, $J_{i,i+L} > 1/(L+1)$ and $J_{i,i-1}<L/(L+1)$, which implies the expected increment of this random walk is positive for each state. In addition, 

    \begin{align*}
        \lim_{i\to-\infty} J_{i,i+L} \geq \frac{\alpha^{L+1}}{L+\alpha^{L+1}} > \frac{1}{L+1}; \lim_{i\to-\infty} J_{i,i-1} \leq \frac{L}{L+\alpha^{L+1}}< \frac{L}{L+1}.
    \end{align*}

    \noindent This implies for all $S \in \mathbb{Z}$, the infimum of expected increments to the left of $S$ is strictly positive, that is, 

    \begin{align*}
        \inf_{i< S} \{LJ_{i,i+L} - J_{i,i-1}\} > 0.
    \end{align*}

    We prove $\lim_{n\to\infty}\mathbb{E}(Y_{n})=\infty$ by Fatou's lemma. To show $\liminf_{n\to\infty}Y_{n} = \infty$, we impose an absorbing state $S>Y_{0}$ such that every state to the right of $S$ collapses into state $S$. Denote the process with this absorbing barrier as $(Y^{(S)}_{n})$, then 

    \begin{align*}
        \lim_{S\to\infty}\liminf_{n\to\infty}Y^{(S)}_{n} &= \sup_{S>0}\sup_{N \geq 1}\inf_{n\geq N}Y^{(S)}_{n} \\
        &= \sup_{N \geq 1}\sup_{S>0}\inf_{n\geq N}Y^{(S)}_{n} \\
        &\leq \sup_{N \geq 1}\inf_{n\geq N}\sup_{S>0}Y^{(S)}_{n} \\
        &= \sup_{N \geq 1}\inf_{n\geq N}Y_{n} \\
        &= \liminf_{n\to\infty}Y_{n}.
    \end{align*}

    \noindent Since the truncated process $(Y^{(S)}_{n})$ is a random walk with an upper absorbing barrier whose expected increments are greater than a positive constant, $\liminf_{n\to\infty}Y^{(S)}_{n} = S$ almost surely. Consequently, $\liminf_{n\to\infty}Y_{n} = \infty$ and we conclude by Fatou's lemma that $\mathbb{E}(Y_{n})\to\infty$ as $n\to\infty$. By convexity of $b(\cdot)$ and Jensen's inequality, $\mathbb{E}(b(Y_{n})) \to \infty$ as $n\to\infty$.

    Notice that $Y_{j(n)} = X_{n}$ where $j(n)$ is the $n$th jump time of the continuous process. By non-explosion, $j(n) \to\infty$ as $n\to\infty$, which implies 

    \begin{align*}
        \lim_{t\to\infty}\mathbb{E}(Y_{t}) = \infty \text{ and } \lim_{t\to\infty}\mathbb{E}(b(Y_{t})) = \infty.
    \end{align*}

\end{proof}


\section{Theorems}
    \begin{theorem}\label{thm1}
        Under the condition of non-explosion in Lemma $\ref{lemma1}$, $i> j$ implies $q_{i}\leq q_{j}$, which further implies $\lim_{i\to \infty}q_{i} = 0$ and $\lim_{i\to-\infty}q_{i} = 1$.
    \end{theorem}
    
    \begin{proof}
        Given the continuous-time process $(\mathbf{Z}_{t})$ is non-explosive, its extinction probability for the continuous-time process $(\mathbf{Z}_{t})$ is the same as that of the discrete embedded branching process $(\mathbf{E}_{n})$ by corollary of Lemma \ref{lemma1}. We use a coupling argument to show the monotonicity of the extinction probabilities. 
        
        Let us fix $i > j$ and construct two continuous-time processes representing cancer populations with different initial types in the following way. Let superscript indicate the initial cell type, that is, $\mathbf{Z}^{(j)}_0 = j$ and $\mathbf{Z}^{(i)}_{0}=i$. The coupling of the process is described as follows. 
        
        At time $0$, five exponentially distributed variables, $B,U,V,D$, and $S$ are competing with each other with rates $b(j),\mu,\nu,d$ and $b(i)-b(j)$, respectively for $(\mathbf{Z}^{(i)}_{t})$ and the minimum of five random variables decides the actions of the initial type $i$ cell. Similarly, there are four competing exponentially distributed random variables $B',U',V'$, and $D'$ with rates $b(j),\mu,\nu$ and $d$ for $(\mathbf{Z}^{(j)}_{t})$ at time $0$. Actions of two initial cells are coupled by setting $B = B', U = U', V=V'$, and $D = D'$. That is, 
        
        \begin{itemize}
            \item If the minimum is $B$, both initial type $i$ cell and initial type $j$ cell proliferate. 
            \item If the minimum is $U$, both initial type $i$ cell and initial type $j$ cell acquire a driver mutation.
            \item If the minimum is $V$, both initial type $i$ cell and initial type $j$ cell acquire a passenger mutation.
            \item If the minimum is $D$, both initial type $i$ cell and initial type $j$ cell die.
            \item If the minimum is $S$, the initial type $i$ cell proliferates and the initial type $j$ cell takes no actions. 
        \end{itemize}
        
        \noindent We further couple cells from two populations after the first action.
        
        \begin{itemize}
            \item If both initial cells proliferate, two new couples are formed, $(i,j)$ and $(i,j)$. 
            \item If both initial cells acquire a driver mutation, the type $i+L$ cell is coupled with the type $j+L$ cell, forming a $(i+L,j+L)$ couple.
            \item If both initial cells acquire a passenger mutation, the type $i-1$ cell is coupled with the type $j-1$ cell, forming a $(i-1,j-1)$ couple.
            \item If both cells die, there will be no new couples.
            \item If the initial type $i$ cell proliferates while the initial type $j$ cell takes no action, one type $i$ cell is coupled with the initial type $j$ cell, forming a $(i,j)$ couple.  
        \end{itemize}
        
        \noindent Notice that in the last scenario, the uncoupled type $i$ cell will evolve (proliferate, mutate, or die) freely, independent of the $(\mathbf{Z}^{(j)}_{t})$ population. Therefore, if we continue this construction, we have $||\mathbf{Z}^{(i)}_{t}||_{\ell^{1}} \geq ||\mathbf{Z}^{(j)}_{t}||_{\ell^{1}}$ for all $t \geq 0$. This implies $q_{i} \leq q_{j}$ and it then follows that $\lim_{i\to\infty}q_{i}=q_{\infty}$ as well as $\lim_{i\to-\infty}q_{i}=q_{-\infty}$ exist and 
        
        \begin{align*}
            &q_{\infty} = q_{\infty}^{2} \Rightarrow q_{\infty} \in \{0,1\} \\
            &q_{-\infty} = \frac{d}{\mu+\nu+d}+\frac{\nu}{\mu+\nu+d}q_{-\infty}+\frac{\mu}{\mu+\nu+d}q_{-\infty} \Rightarrow q_{-\infty} = 1.
        \end{align*}
        
        \noindent If $q_{\infty} = 1$, we have $q_{i} = 0$ for all $i\in\mathbb{Z}$ and this contradicts the fact that all extinction probabilities are strictly less than one. Therefore, $\lim_{i\to\infty}q_{i} = 0$. 
    \end{proof}


\begin{theorem}\label{thm2}
    When $\frac{\nu}{\mu L} \leq 1$, $(Y_{n})$ and $(X_{t})$ are transient.
\end{theorem}
    
\begin{proof}
    It suffices to show that there exists a $1-$subinvariant measure that is not invariant by Theorem 5.4 in \cite{seneta2006non}. We claim that $\mathbf{x} = (\delta(n)(1-q_{n})(1-T_{n,n}))_{n\in\mathbb{Z}}$ is the desired measure.
        
    \begin{align*}
        \mathbf{x}^{T}J \lneqq \mathbf{x}^{T} \iff &\mu(1-q_{i})+\nu(1-q_{i}) \leq \delta(i)(1-q_{i})(1-T_{i,i}) \\
        \iff &\frac{\mu}{\delta(i)}+\frac{b(i)}{\delta(i)}\frac{1-q_{i}^{2}}{1-q_{i}}+\frac{\nu}{\delta(i)} \leq 1 \\
        \iff &b(i)q_{i} \leq d \\
        \iff &q_{i} \leq \frac{d}{b(i)}.
    \end{align*}
	    
    \noindent Equality cannot hold for all $i$ since $\frac{d}{b(i)} > 1$ for small $i$. Define $\mathbf{q}'= (\min(1,\frac{d}{b(i)}))_{i=-\infty}^{\infty}$. To show $q_{i} \leq \frac{d}{b(i)}$, is suffices to show $\mathbf{q}' \geq \mathbf{G}(\mathbf{q}')$. To see this, notice that
	    
    \begin{align*}
        &\mathbf{q}' \geq \mathbf{G}(\mathbf{q}') \Rightarrow \mathbf{G}(\mathbf{q}') \geq \mathbf{G}^{\circ 2}(\mathbf{q}') \Rightarrow \cdots \Rightarrow \mathbf{G}^{\circ n}(\mathbf{q}') \geq \mathbf{G}^{\circ (n+1)}(\mathbf{q}') \\
        \Rightarrow&\;\mathbf{q}' \geq \lim_{n\to\infty}\mathbf{G}^{\circ n}(\mathbf{q}') \;\geq\; \mathbf{q}, \text{ since } \mathbf{q} \text{ is the minimal fixed point,}
    \end{align*}
    
    \noindent and $\mathbf{G}^{\circ n}$ is the $n$-fold composition of $\mathbf{G}$ with itself. Now, we verify $\mathbf{q}' \geq \mathbf{G}(\mathbf{q}')$. When all three quantities ($q'_{i-1},q'_{i},q'_{i+L}$) are less than or equal to $1$,
	    
    \begin{align*}
        \mathbf{q}' \geq \mathbf{G}(\mathbf{q}') 
        \iff &\frac{d}{b(i)} \geq \frac{d}{\delta(i)}+\frac{\nu}{\delta(i)}\frac{d}{b(i-1)}+\frac{b(i)}{\delta(i)}(\frac{d}{b(i)})^{2}+\frac{\mu}{\delta(i)}\frac{d}{b(i+L)} \\
        \iff &d\delta(i) \geq db(i) + \nu d(1+s_{p}) + d + \mu d(1+s_{p})^{-L} \\
        \iff &\mu (1-(1+s_{p})^{-L}) \geq \nu s_{p}.
    \end{align*}
	    
    \noindent Since $(1+s_{p})^{-L} \geq 1-Ls_{p}$, we have 
	    
    \begin{align*}
        \mu L \geq \nu \Rightarrow &\mu Ls_{p} \geq \nu s_{p} \\
        \Rightarrow &\mu Ls_{p} \geq \nu s_{p} \\
        \Rightarrow &\mu (1-(1+s_{p})^{-L}) \geq \nu s_{p}.
    \end{align*}
	    
    \noindent If $\frac{d}{b(i-1)} > 1$ and $\frac{d}{b(i)} \leq 1$, by the same computation, we have
	    
    \begin{align*}
        \frac{d}{b(i)} &\geq \frac{d}{\delta(i)}+\frac{\nu}{\delta(i)}\frac{d}{b(i-1)}+\frac{b(i)}{\delta(i)}(\frac{d}{b(i)})^{2}+\frac{\mu}{\delta(i)}\frac{d}{b(i+L)} \\
	&\geq \frac{d}{\delta(i)} +\frac{\nu}{\delta(i)}+\frac{b(i)}{\delta(i)}(\frac{d}{b(i)})^{2}+\frac{\mu}{\delta(i)}\frac{d}{b(i+L)} \\
        &=\frac{d}{\delta(i)} +\frac{\nu}{\delta(i)}q_{i-1}'+\frac{b(i)}{\delta(i)}q_{i}'^{2}+\frac{\mu}{\delta(i)}q_{i+L}'.
    \end{align*}
	    
    \noindent If $\frac{d}{b(i-1)} > 1$ and $\frac{d}{b(i)} > 1$, $q'_{i} = q'_{i-1} = 1$ and the inequality holds trivially. Since the process is non-explosive, $(Y_{n})$ is transient implies $(X_{t})$ is transient as well (Theorem $3.4.1$ in $\cite{norris1998markov}$).
\end{proof}


\begin{theorem}\label{thm3}
    When $\frac{\nu}{\mu L} > 1$, the jump chain $(Y_{n})$ is positive recurrent. Consequently, $(X_{t})$ admits a limiting distribution.
\end{theorem}
    
\begin{proof}
    It is clear that the jump chain is irreducible with period $L+1$. To show it is positive recurrent, we prove by contradiction. Suppose the jump chain is transient or null recurrent, then by Theorem $1.8.5$ in \cite{norris1998markov}, for any $x \in \mathbb{Z}_{+}$,  
        
    \begin{align*}
         \lim_{n\to\infty} \mathbb{P}(|Y_{n}|>x)=1-\lim_{n\to\infty} \mathbb{P}(|Y_{n}|\leq x)= 1-\lim_{n\to\infty}\sum_{k=-x}^{x}\mathbb{P}(|Y_{n}|= k) = 1.
    \end{align*}
        
    \noindent To show it is not the case, we bound the right and left tail probabilities $(i\to\infty \text{ and } i\to-\infty)$. 
    
    Observe that as $i\to\infty$, the jump chain tends to a random walk with 
        
    \begin{align*}
        p_{i,i-1} = \frac{\nu}{\mu+\nu}>\frac{L}{L+1}, p_{i,i+L} = \frac{\mu}{\mu+\nu}<\frac{1}{L+1}.
    \end{align*}
        
    \noindent Therefore, there exists $I_{r}>0$ such that for all $i \geq I_{r}$, $p_{i,i-1} > q \in (\frac{L}{L+1},\frac{\nu}{\mu+\nu})$ and $p_{i,i+L} < p \in (\frac{\mu}{\mu+\nu},\frac{1}{L+1})$ and $p+q = 1$. Let $(R_{n})$ be a random walk with a retaining barrier starting from $R_{0}=I_{r}$ and for $i\geq I_{r}$,
        
    \begin{align*}
        p_{i,\max\{I_{r},i-1\}} = q, p_{i,i+L} = p.
    \end{align*}
        
    \noindent This random walk is more likely to increase than the original process for $i \geq I_{r}$. Since this random walk is more likely to decrease by $L$ than increase by $L$, $(R_{n})$ is positive recurrent. To see this,
        
    \begin{align*}
        q^{L} \geq (\frac{L}{L+1})^{L} > \frac{1}{L+1} \geq p.
    \end{align*}
        
    As already stated, by Lemma $\ref{lemma2}$, as $i \to -\infty$, the jump chain approaches a random walk with 
        
    \begin{align*}
        p_{i, i-1} = \frac{\nu}{\mu+\nu+d}\alpha^{-1} < \frac{L}{L+1}, p_{i,i+L} = \frac{\mu}{\mu+\nu+d}\alpha^{L} > \frac{1}{L+1}.
    \end{align*}
        
    \noindent Again, there exists $I_{l}<0$ such that for all $j\leq I_{l}$, $p_{j, j-1} < q \in (\frac{\nu}{\mu+\nu+d}\alpha^{-1},\frac{L}{L+1})$ and $p_{j,j+L} > p \in (\frac{1}{L+1},\frac{\mu}{\mu+\nu+d}\alpha^{L})$ and $p+q = 1$. Let $(L_{n})$ be a random walk starting from $I_{l}$ with a retaining barrier on the right, and take $q \in (\frac{\nu}{\mu+\nu+d}\alpha^{-1},\frac{L}{L+1}), p \in (\frac{1}{L+1},\frac{\mu}{\mu+\nu+d}\alpha^{L})$ such that for all $j \leq I_{l}$,
        
    \begin{align*}
        p_{j,j-1} = q, p_{j,\min\{I_{l},j+L\}} = p.
    \end{align*}
        
    \noindent This is a left-continuous random walk with a retaining barrier whose expected increment $Lp-q$ is greater than $0$. Therefore, by \cite{spitzer2001principles} Case $(i)$ on page $188$ and $P1$ on page $191$, the process is positive-recurrent. 
        
    Fix $\epsilon \in (0,1)$; then there exist $u>I_{r}>0$ and $0>I_{l}>v$ such that 
        
    \begin{align*}
        \lim_{n\to\infty}\mathbb{P}(R_{n}>u) < \frac{\epsilon}{2} \text{ and } \lim_{n\to\infty}\mathbb{P}(L_{n}<v) < \frac{\epsilon}{2}.
    \end{align*}
        
    \noindent Since $(R_{n})$ (resp. $(L_{n})$) has a heavier right (resp. left) tail than the jump chain, 
        
    \begin{align*}
        \limsup_{n\to\infty} \mathbb{P}(Y_{n}>u) < \frac{\epsilon}{2} \text{ and } \limsup_{n\to\infty} \mathbb{P}(Y_{n}<v) < \frac{\epsilon}{2}.
    \end{align*}
        
    \noindent Therefore, 
        
    \begin{align*}
        \limsup_{n\to\infty} \mathbb{P}(|Y_{n}|>\max\{u,|v|\}) \leq \limsup_{n\to\infty} \mathbb{P}(R_{n}>u) + \limsup_{n\to\infty} \mathbb{P}(L_{n}<v) < \epsilon<1.
    \end{align*}
        
    \noindent This is a contradiction, so the jump chain must be positive recurrent. Consequently, it admits a unique invariant distribution $\mathbf{y}$. That is, for all $i\in\mathbb{Z}$,
        
    \begin{align*}
        y_{i-L}J_{i-L,i} + y_{i+1}J_{i+1,i} = y_{i} \text{ and } \sum_{i\in\mathbb{Z}} y_{i} = 1.
    \end{align*}

    Recall that the holding time for state $i$ follows an exponential distribution with $\delta(i)(1-T_{i,i})$, and we have 

    \begin{align*}
    &\begin{cases}
        \lim_{i\to\infty}\delta(i)(1-T_{i,i}) = \lim_{i\to\infty}\mu\frac{1-q_{i+L}}{1-q_{i}}+\nu\frac{1-q_{i-1}}{1-q_{i}} = \mu+\nu \\
        \lim_{i\to-\infty}\delta(i)(1-T_{i,i}) = \mu\alpha^{L}+\nu\alpha^{-1} = \mu+\nu+d
    \end{cases} \\
    \Rightarrow &\inf_{i\in\mathbb{Z}}\{\delta(i)(1-T_{i,i})\} > 0.
    \end{align*}

    \noindent Since the infimum of the rates is strictly greater than $0$, $y_{i}/(\delta(i)(1-T_{i,i}))$ a summable invariant measure for the generator $Q$, which concludes that $(X_{t})$ admits a limiting distribution (Theorem $3.6.2$ of \cite{norris1998markov}).  

\end{proof}

\section{Discrete Version of Levinson’s Fundamental Theorem} \label{Levinson}

Let $y(n)\in \mathbb{C}^{p}$ and $y(n+1) = [\Lambda(n)+R(n)]y(n)$, where $\Lambda(n) = \text{diag}(\lambda_{1}(n),\cdots,\lambda_{p}(n))$. Suppose $\Lambda(n)$ is invertible for $n\geq n_{0}$ and following conditions are satisfied. Since any two matrix norms are equivalent, we use the operator norm, denoted $||\cdot||$.

\begin{align*}
    &\exists K_{1},K_{2} \text{ s.t. } \forall i\neq j \text{ either }
    \begin{cases}
    &[\prod_{k=n_{0}}^{n}|\frac{\lambda_{j}(k)}{\lambda_{i}(k)}| \to 0 \text{ as } n\to\infty \text{ and } \\
    &\qquad \prod_{k=n_{n_{1}}}^{n_{2}}|\frac{\lambda_{j}(k)}{\lambda_{i}(k)}| \leq K_{1}, \forall n_{0}\leq n_{1}\leq n_{2}] \\
    &\text{or } \prod_{k=n_{n_{1}}}^{n_{2}}|\frac{\lambda_{j}(k)}{\lambda_{i}(k)}| \geq K_{2}, \forall n_{0}\leq n_{1}\leq n_{2},
    \end{cases} \\
    &\text{and }\sup_{1\leq i\leq p}\sum_{n=n_{0}}^{\infty}\frac{||R(n)||}{|\lambda_{i}(n)|} < \infty.
\end{align*}

\noindent Then, the perturbed linear system $y(n+1) = [\Lambda(n)+R(n)]y(n)$ has a fundamental matrix satisfying 

\begin{align*}
    Y(n) = (I + o(1))\prod_{k=n_{0}}^{n-1}\Lambda(k) \text{ as } n\to\infty.
\end{align*}

\noindent Columns of $Y(n)$ are linearly independent solution vectors such that $y(n) = Y(n)y(0)$.

\clearpage

\bibliography{mybib}

\end{document}